\documentclass[
	draft=false,a4paper
	,11pt
	,DIV=9
	,parskip=half*
	,bibliography=totocnumbered
	]{scrartcl}

\pdfoutput=1		
 
\usepackage{ucs}				
\usepackage[utf8x]{inputenc} 
\usepackage[T1]{fontenc}
\usepackage[USenglish]{babel}
 
\usepackage{graphicx}
\usepackage[pdfusetitle,backref=page]{hyperref}		
 
\usepackage{pdf14} 
\usepackage{amsmath}
\usepackage{amssymb}		
\usepackage{amsthm}
\usepackage{enumerate}
\usepackage{float}	

\usepackage[numbers]{natbib}
\usepackage{url}

\usepackage{algorithm}			
\usepackage
	{algorithmic}

\usepackage{tikz}
\usepackage{tkz-graph}
\usepackage{tkz-berge}
\usepackage{multirow}

\hypersetup{
    pdfinfo={
        Keywords={combinatorial optimization, robust optimization, perfect matching, bipartite graphs}
    }
}

\theoremstyle{plain}
\newtheorem{theorem}{Theorem}
\newtheorem*{theorem*}{Theorem}

\newtheorem{corollary}[theorem]{Corollary}
\newtheorem{lemma}[theorem]{Lemma}

\newtheorem{proposition}[theorem]{Proposition}

\theoremstyle{definition}

\newtheorem*{problem*}{Problem}
 
\theoremstyle{remark}

\newtheorem*{remark*}{Remark}

\newcommand*{\sepsign}{\;{\footnotesize $\blacklozenge$}\; }
\DeclareMathOperator{\dcup}{\dot{\cup}}
\DeclareMathOperator{\neighb}{N}	
\newcommand*{\R}{\mathbf{R}}
\newcommand*{\Z}{\mathbf{Z}}
\newcommand*{\mS}{\mathcal{S}}
\newcommand*{\mC}{\mathcal{C}}
\renewcommand*{\bf}[1]{\textbf{#1}}
 
\newcommand*{\formatmathnames}[1]{\textnormal{\small #1}}
\newcommand*{\vrpm}{\formatmathnames{V-RAP}}  
\newcommand*{\uvrpm}{\formatmathnames{card-V-RAP}} 
\newcommand*{\rap}{\formatmathnames{E-RAP}} 
\newcommand*{\urap}{\formatmathnames{card-E-RAP}} 
\newcommand*{\vcthree}{\formatmathnames{VC\,3}}
\newcommand*{\setcover}{\formatmathnames{SC}}
\newcommand*{\p}{\formatmathnames{P}}
\newcommand*{\np}{\formatmathnames{NP}}
\newcommand*{\dtime}{{\formatmathnames{DTIME}}}

\newcommand*{\ptas}{\formatmathnames{PTAS}}
\newcommand*{\lp}{\formatmathnames{LP}}
\newcommand*{\ilp}{\formatmathnames{ILP}}
\newcommand*{\opt}{\formatmathnames{OPT}}

\newcommand*{\uset}{\mathcal{F}} 
\newcommand*{\rset}{X} 
\newcommand*{\inst}{\mathcal{I}} 
\newcommand*{\fset}{f} 

\author{
  \texorpdfstring		
  {\large%
    David Adjiashvili,$^{\textrm{a}}$
    \quad 
    Viktor Bindewald,$^{\textrm{b}}$
    \quad
    Dennis Michaels$^{\textrm{b}}$}
  {David Adjiashvili, Viktor Bindewald, Dennis Michaels}
}

\publishers{\centering 
  \begin{minipage}{0.98\textwidth}
    \small
    \begin{itemize}
      \setlength{\itemsep}{0ex}
    \item[$^{\textrm{a}}$] Institute for Operations Research,
      Department of Mathematics, ETH 
      Zurich, R\"amistrasse 101, 8092 Zurich (Switzerland),\\ 
      addavid@ethz.ch
    \item[$^{\textrm{b}}$] Department of Mathematics, TU Dortmund University,
      Vogelpothsweg 87, 44227 Dortmund (Germany),\\ \{viktor.bindewald, dennis.michaels\}@math.tu-dortmund.de
    \end{itemize}
  \end{minipage}
}

\title{\Large%
  Robust Assignments with Vulnerable Nodes
}

\date{\normalsize%
version from: \today\\
\text{} 
}

\setkomafont{title}{\normalfont\tiny\bfseries}
\setkomafont{section}{\normalfont\large\bfseries}
\setkomafont{subsection}{\normalfont\normalsize\bfseries}
\setkomafont{paragraph}{\normalfont\normalsize\itshape}

\begin{document} 
    
\maketitle  

\begin{abstract}
Various real-life planning problems require making upfront decisions
before all parameters of the problem have been disclosed. An important
special case of such problem especially arises in scheduling and staff
rostering problems, where a set of tasks needs to be assigned to an available
set of resources (personnel or machines), in a way that each task is
assigned to one resource, while no task is allowed to share a resource
with another task. In its nominal form, the resulting computational 
problem reduces to the well-known assignment problem that can be
modeled as matching problems on bipartite graphs. 

In recent work \cite{adjiashvili_bindewald_michaels_icalp2016}, a new
robust model for the assignment problem was introduced that can
deal with situations in which certain resources, i.e.\ nodes or edges of
the underlying bipartite graph, are vulnerable and may become
unavailable after a solution has been chosen. 
In the original version from \cite{adjiashvili_bindewald_michaels_icalp2016} the resources subject to uncertainty are the
edges of the underlying bipartite graph.  

In this follow-up work, we complement our previous study 
by considering nodes as being vulnerable, instead of edges. 
The goal is now to choose a minimum-cost collection of nodes such
that, if any vulnerable node becomes unavailable, the remaining
part of the solution still contains sufficient nodes to perform all
tasks.
From a practical point of view, such type of unavailability is
interesting as it is typically caused e.g.\ by an employee's
sickness, or machine failure. 
We present algorithms and hardness of approximation results for
several variants of the problem.

\end{abstract}

\textbf{Keywords.\quad}
assignment problems \sepsign structural
robustness \sepsign
combinatorial optimization \sepsign robust optimization \sepsign approximation algorithms

\section{Introduction}
\label{sec:intro}

Input data of optimization problems is often uncertain in practice. 
A classical approach for dealing with uncertainty is robust optimization,
an approach that seeks solutions performing well in \emph{worst case} 
realizations of the uncertain parameters of the feasible set and the
cost structure.

In our former work \cite{adjiashvili_bindewald_michaels_icalp2016} a new 
version of Robust Assignment Problems was introduced and 
studied\footnote{A full version of this conference contribution 
	including all proofs is given in the preprint~\cite{icalp_full_version_16}.}. 
The robust version studied therein is a bulk-robust problem. This
robustness framework was introduced in Adjiashvili et
al.~\cite{AdjishviliStillerZenklusen2015}, and follows the general
idea of redundancy-based robustness.   
Roughly speaking, bulk-robustness deals with combinatorial
optimization problems in which some resources (e.g.\ edges or nodes in
graphs) may become unavailable. 
This type of uncertainty is modelled
by listing all subsets of resources that can simultaneously become
unavailable. 
The union of all these subsets form the set of so-called
\emph{vulnerable} resources. 
Each member of the list then defines a
scenario that leads to a deletion of the corresponding vulnerable
resources.  In this setting, a \emph{robust solution}
is a subset of resources that provides a feasible solution to the
underlying nominal optimization problem in every scenario, and the
task is to determine a robust solution of minimum cost. 

For the Robust Assignment Problem considered in
\cite{adjiashvili_bindewald_michaels_icalp2016} the nominal
optimization problem is given by the bipartite perfect matching
problem while vulnerable resources are edges. 
We call this problem \emph{Edge-Robust Assignment
	Problem} (\rap).
The formal definition of \rap\  is the following.
\begin{problem*}[The Edge-Robust Assignment Problem (\rap)]
	\text{}  
	\begin{itemize}
		\item \underline{Input:} 
		Tuple $(G,\uset,c)$, where $G := (U \dcup W,E)$ is a balanced,
		bipartite graph, i.e.\  $|U|=|W|$, $\uset \subseteq 2^E$ is a family
		of sets of vulnerable edges, and $c\in\R_{\geq 0}^E$ is a non-negative
		cost vector.  
		\item 
		\underline{Output:} If existent, an optimal solution to
		\begin{align}
		\tag{\rap}
		\begin{array}{ll}
		\min & c(\rset)  \\ 
		\textnormal{s.t.} & 
		\forall F\in \uset: 
		\rset \setminus F \text{ contains a perfect matching
			in } G\\
		& \rset \subseteq E .
		\end{array}
		\end{align}
	\end{itemize}
\end{problem*}

The focus in~\cite{adjiashvili_bindewald_michaels_icalp2016} is on the case
where $|F| = 1$ for all $F\in \uset$.
In this follow-up work, we complement our analysis for Robust Assignment Problems
by now considering nodes as vulnerable resources. 
For this, we again consider the nominal version of the problem to be the
matching problem on a bipartite graph $G=(U\dcup W,E)$, in which the goal
is to match all nodes in one side of the bipartition, say $U$. In other 
words, the nominal goal is to find a $U$-perfect matching.
Also, we assume that all vulnerable nodes are contained in the other
side of the bipartition, namely $W$, an assumption that is very natural the 
applications motivating this work. The nodes in $U$ and $W$ can naturally be
interpreted as jobs and machines, respectively, and our assumption on the 
failure model implies that some machines may fail.
Finally, we assume that the costs are associated with machines, i.e.\ with the 
nodes in $W$. Once a node $w\in W$ is included, all edges in $G$ connecting
it to nodes in $U$ are also automatically included.
Thus, a \emph{robust solution} is a node subset $W' \subseteq W$ such
that in any scenario $F$ (corresponding to deletion of some nodes in $W$)
the remaining nodes in $W' \setminus F$ suffice to construct a $U$-perfect matching,
and the goal is to find a robust solution of minimum cost.
We call this problem \emph{Node-Robust Assignment Problem}, and denote it by \vrpm.
Figure~\ref{fig:v-rap-example-and-its-solution} shows an example of a \vrpm\ instance.
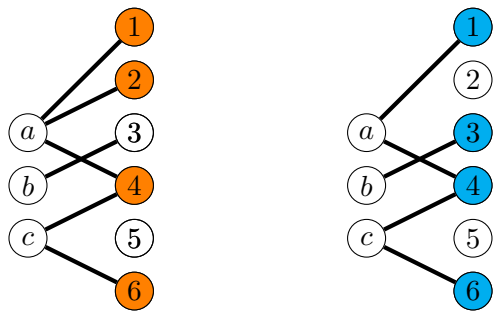
\begin{figure}[H]
	\label{fig:v-rap-example-and-its-solution}
	\centering
	\begin{minipage}{0.2\linewidth}
		\begin{tikzpicture}[scale=0.7, every node/.style={circle, draw, minimum size=0.5cm, inner sep = 0pt}]
		\node (a) at (0,1) [] {$a$};
		\node (b) at (0,0) [] {$b$};
		\node (c) at (0,-1) [] {$c$};
		
		\node (1) at (2,3) [] {$1$};
		\node (2) at (2,2) [] {$2$};
		\node (3) at (2,1) [] {$3$};
		\node (4) at (2,0) [] {$4$};
		\node (5) at (2,-1) [] {$5$};
		\node (6) at (2,-2) [] {$6$};

		\node (1) at (2,3) [fill=orange] {$1$};
		\node (2) at (2,2) [fill=orange] {$2$};
		\node (3) at (2,1) [] {$3$};
		\node (4) at (2,0) [fill=orange] {$4$};
		\node (5) at (2,-1) [] {$5$};
		\node (6) at (2,-2) [fill=orange] {$6$};

		\draw[-, ultra thick] (a) -- (1);
		\draw[-, ultra thick] (a) -- (4);
		\draw[-, ultra thick] (a) -- (2);
		\draw[-, ultra thick] (b) -- (3);
		\draw[-, ultra thick] (c) -- (4);
		\draw[-, ultra thick] (c) -- (6);  
		\end{tikzpicture}
	\end{minipage}
	\hspace{0.1\linewidth}
	\begin{minipage}{0.2\linewidth}
		\begin{tikzpicture}[scale=0.7, every node/.style={circle, draw, minimum size=0.5cm, inner sep = 0pt}]
		\node (a) at (0,1) [] {$a$};
		\node (b) at (0,0) [] {$b$};
		\node (c) at (0,-1) [] {$c$};
		
		\node (1) at (2,3) [] {$1$};
		\node (2) at (2,2) [] {$2$};
		\node (3) at (2,1) [] {$3$};
		\node (4) at (2,0) [] {$4$};
		\node (5) at (2,-1) [] {$5$};
		\node (6) at (2,-2) [] {$6$};	
		
		\begin{scope}
		\node (3) at (2,1) [fill=cyan] {$3$};
		\draw[-, ultra thick] (b) -- (3);
		\end{scope}			
		
		\begin{scope}
		\node (4) at (2,0) [fill=cyan] {$4$};
		\draw[-, ultra thick] (a) -- (4);
		\draw[-, ultra thick] (c) -- (4);
		\end{scope}			
		
		\begin{scope}
		\node (6) at (2,-2) [fill=cyan] {$6$};
		\draw[-, ultra thick] (c) -- (6);
		\end{scope}

		\begin{scope}
		\node (1) at (2,3) [fill=cyan] {$1$};
		\draw[-, ultra thick] (a) -- (1);
		\end{scope}			  
		\end{tikzpicture}
	\end{minipage}
	\caption{left: A \vrpm\ instance with four scenarios given by $\{1\}$, $\{2\}$,
		$\{4\}$, $\{6\}$ 
		(orange nodes);\newline
		right: A node-robust solution (cyan nodes)
		of the \vrpm\ instance. }
\end{figure}

A different view on \vrpm\ is the following.
The decision maker has to select a set of machines such that all jobs can be performed.
Ahead of this decision his adversary announces a list of sets of machines
from which he will sabotage one set of machines after the decision is made. 
The decision maker has to make sure that the act of sabotage does not
jeopardize the completion of all jobs on the selected machines. 

Before formally stating \vrpm, we briefly present two concrete motivating 
applications exhibiting its practical relevance.
\begin{itemize}
	\item
	\emph{Staff Scheduling and Rostering. }
	Rostering is an important task of Human Resource Management.
	The rostering process has four different aspects: strategic, tactical, operational and retrospective
	(for details we refer the reader to \cite{ernst_et_al_04}).
	\vrpm\ addresses the tactical dimension of rostering.
	The objective thereby is to deliver a roster for a certain planning horizon (e.g.\ a month), which can be used as
	a starting point for operative assignment decisions on a daily (or weekly) basis, resulting in a schedule.
	Naturally, a roster is subject to uncertainty because employees may
	get sick or not be able to work in a certain shift due to legal work
	load regulations or non-availability of crucial equipment. 
	Uncertainty of this kind can be modeled in the bulk-robust framework.
	A solution to \vrpm\ enables the decision maker to choose an
	assignment from the solution for the upcoming day (or week), meeting
	the latest needs of both the company and the personnel.  
	In the context of health care, for instance, \vrpm\ can be used to address
	continuity of service, maintaining full staff levels and minimizing
	costs simultaneously. 
	\item 
	\emph{Subcontracting. } 
	In a many companies numerous tasks are outsourced to subcontractors or freelancers.
	A typical example is the design and deployment of a new IT infrastructure. 
	During the evaluation process a redundant set of potential subcontractors is considered.
	It is expected that some of the candidates will not be available
	after the decision process due to various reasons such as commitment
	elsewhere in the desired period of time or sickness of key team
	members. It is hence important to contact and negotiate with a redundant
	set of subcontractors to later be able to accommodate the latter kind of
	unavailabilities. \vrpm\ can be used to decrease the costs of the negotiation
	and hiring process.
\end{itemize}

\bigskip
The formal setting of the Node-Robust Assignment Problem as studied
in this work is the following.
We are given a simple bipartite graph $G=(U \dcup W, E)$ with
$|U| \leq |W|$. 
Elements in $U$ are associated with jobs while elements from $W$
represent machines. In addition, each machine node $w$ is associated
with a non-negative cost $c_w\in\R_{\geq 0}$, yielding a cost vector
$c\in\R^W_{\geq 0}$.
A set $\rset \subseteq W$ of machine nodes is called an
\emph{assignment} of $G$ if the subgraph $G[U\dcup \rset]$
induced by the selected machine nodes in $\rset$ and all job nodes
from $U$ contains a $U$-perfect matching, 
i.e.\ a set of non-adjacent edges incident to all job nodes in $U$.
In the nominal version the objective is to find an assignment 
of $G$ with minimum cost. 
In the node-robust version, the set $W$ of machine
nodes is subject to uncertainty. 
The uncertainty is described by a list $\uset \subseteq 2^W$ of subsets of machine
nodes. Each element $F$ of $\uset$ defines a failure scenario, the 
occurrence of which involves the removal of all machine nodes
in $F$ from $G$. We are now interested in determining a subset
$\rset\subseteq W$ of machine nodes in such a way that,
for every failure scenario $F\in\uset$, the induced subgraph $G[U\dcup
\rset]-F$ contains a $U$-perfect matching. We call such a
subset a \emph{node-robust solution}. The aim of
\vrpm\ is to find, for a bipartite graph $G$ and for a set $\uset$ of
failure scenarios, a node-robust solution that is minimal 
with respect to the cost vector $c$.   

We will restrict ourselves in this manuscript to the case when each failure
scenario is defined by one single machine node, i.e.\ we assume that $\mathcal{F}$
is a family of singletons, and hence can simply be represented a subset
of $W$, namely $\uset \subseteq W$.
As shown in our previous work for \rap, we will prove in this work that
the node-robust version exhibits interesting structure under this
assumption. In particular, we will show that this case is hard to
approximate within a sub-logarithmic factor in general, and within 
a constant grater than one in the unweighted case. We also provide
approximation algorithms with matching approximation guarantees (up 
to small additive constants). 

The Node-Robust Assignment Problem under consideration is formally 
stated next.  
\begin{problem*}[The Node-Robust Assignment Problem (\vrpm)]
	\text{}
	\begin{itemize}
		\item \underline{Input:} 
		Tuple $(G,\uset,c)$, where $G := (U \dcup W,E)$ is a bipartite 
		graph, $\uset \subseteq W$ is a set of vulnerable nodes, and
		$c\in\R_{\geq 0}^W$ is a non-negative cost vector.   
		\item 
		\underline{Output:} If existent, an optimal solution to 
		\begin{align}
		\tag{\vrpm}
		\begin{array}{ll}
		\min & c(\rset)  \\ 
		\textnormal{s.t.} & 
		\forall \fset\in \uset:\;\; 
		G[U\dcup \rset] - \{f\}
		\text{ contains a $U$-perfect matching}, \\
		& \rset \subseteq W .
		\end{array}
		\end{align}
	\end{itemize}
\end{problem*} 
There are two important cases of \vrpm\ that are worth mentioning.
The first case concerns the structure of the uncertainty set.  
A \vrpm\ instance $((U\dcup W, E),\uset,c)$ is called \emph{uniform}
if each machine node is vulnerable, i.e.\ $\uset=W$. 
The second case is that of unit costs, i.e.\ the case $c_w=1$, for each $w\in W$. In this
case, one is interested in finding a node-robust solution with minimum
cardinality. We will denote this case by \uvrpm.

We remark that it is easy to verify whether a given instance $(G,\uset,c)$ of \vrpm\ is feasible. 
This can be achieved by applying any efficient bipartite maximum
matching algorithm $|\uset| \leq |W|$ times. 
Therefore, we will assume that every \vrpm\ instance
considered in this work is feasible.

\bigskip
Analogously to~\cite{adjiashvili_bindewald_michaels_icalp2016} for the
Edge-Robust Assignment Problem,  
we focus in this work on deriving hardness results and approximation
algorithms for \vrpm\ and \uvrpm. 
Our results are summarized in Table~\ref{tab:result-for-variants-of-the-RAP-problem}. This table
also compares our results for \vrpm\ with our results for \rap. 
\begin{table}[H]
	\label{tab:result-for-variants-of-the-RAP-problem}
	\centering
	\begin{tabular}{|c|c|c|}
		\hline
		Problem & hardness & algorithm's \\
		($\uset \subseteq W$, $E$) & of approximation & guarantee\\[1.5ex]
		\hline
		\vrpm & $\boldsymbol{d \log n}$, $d < 1$
		[Thm.~\ref{thm:hardness-of-V-RAP}] & $\boldsymbol{\log n +
			2}$
		[Thm.~\ref{thm:v-rap-log-n-plus-one-algorihm}]\\[1.5ex]  
		\uvrpm &  \textbf{no PTAS}
		[Thm.~\ref{thm:card-v-rap-is-hard-to-approximate}] &
		$\mathbf{1.75}$ [Thm.~\ref{thm:card-v-rap-admits-a-1.75-approximation}] \\ 
		& & ($\mathbf{2}$ on non-uniform instances) \\[1.3ex] 
		\hline 
		\rap & $\boldsymbol{d \log n}$, $d < 1$
		\cite[Thm.~3]{adjiashvili_bindewald_michaels_icalp2016} & $\boldsymbol{O(\log
			n)}$ (randomized)
		\cite[Thm.~4]{adjiashvili_bindewald_michaels_icalp2016}\\[1.5ex]  
		\urap$^{\P}$ & \textbf{no PTAS}
		\cite[Thm.~5]{adjiashvili_bindewald_michaels_icalp2016} &
		$\mathbf{1.5}$
		\cite[Thm.\,6]{adjiashvili_bindewald_michaels_icalp2016}
		\\  
		& & ($\mathbf{3}$ on non-uniform instances) \\
		\hline
	\end{tabular} 
	\caption{Summary of results for \vrpm\ as will be presented in this work and
		for \rap\ as stated in \cite{adjiashvili_bindewald_michaels_icalp2016} (for proofs, see the full preprint
		version \cite{icalp_full_version_16}) where $n$ represents the number of nodes in the underlying bipartite graph. \newline
		{\footnotesize $^\P$ \urap\ denotes the unweighted Edge-Robust Assignment Problem.}
	} 
\end{table}

\bigskip
The paper is organized as follows.
In Section~\ref{sec:relatedwork} we briefly review existing results
related to \vrpm\ (for related work concerning \rap\ see \cite{adjiashvili_bindewald_michaels_icalp2016}).
Section~\ref{sec:hardness-of-V-RAP} deals with the general, weighted
version of \vrpm, for which we show hardness of approximation 
and present an approximation algorithm with an approximation.
The bounds depend logarithmically on the number of nodes. 
In Section~\ref{sec:hardness-of-card-V-RAP} we focus on the
unweighted version \uvrpm. 
We first prove that there cannot be exist a \ptas\ for
\uvrpm, unless $\p=\np$. We then analyze the approximation algorithm
introduced for general \vrpm\ for the unweighted case. 
We show that this algorithm becomes a constant factor
approximation algorithms for \uvrpm. \\
One interesting fact about the Edge-Robust Assignment Problem is that this problem is
even \np-hard in its simplest variant, i.e.\ in case of two
vulnerable edges and unit weights (see \cite{adjiashvili_bindewald_michaels_icalp2016}, and \cite[Section
5]{icalp_full_version_16} for a proof).  
In Section~\ref{sec:cardvrap-twovulnerablemachinecase} we show that,
in contrast to the situation \rap, the case of two vulnerable 
machines nodes is solvable in polynomial time.

\paragraph{Notation.}
Throughout this work we use the following notation.
Let $G$ be a bipartite graph. By $V(G)$ and $E(G)$ we denote the node
set and the edge set of $G$. For a subset $E^\prime\subseteq E(G)$, $V(E^\prime)$ is
used to represent the set of all nodes incident by $E^\prime$. For a subset
$V^\prime \subseteq V(G)$, $E(V^\prime)$ denotes the set of edges
of $G$ with both endpoints in $V^\prime$. The subgraph
$(V^\prime, E(V^\prime))$ induced by a node subset $V^\prime$ is
abbreviated by $G[V^\prime]$.
For a subset $V^\prime\subseteq V(G)$ of nodes, we use the notation
$G-V^\prime$ to denote the graph resulting from $G$ when all
nodes from $V^\prime$ and all edges incident to some node of
$V^\prime$ are removed.
 
\section{Related work}
\label{sec:relatedwork}
 
Laroche et al.~\cite{laroche_et_al_14} considered a problem related to \vrpm.
For a given bipartite graph $G:=(U\dcup W, E)$ they consider the following
interdiction problem: 
Does the removal of $k$ arbitrary nodes from $W$ results in a
graph without a $U$-perfect matching? 
The authors were especially interested in computing the smallest
number $k$ for which the answer to the latter question is yes.
The motivation to study the problem comes from the nurse rostering
problem arising in health care. In that context one is interested in
determining the largest number of nurses that can be absent such that
all patients can still be treated adequately.  
This largest number can be seen as a measure for the resilience of a
health care provider with respect to staff unavailability.    
In the setting of \vrpm, the question above can be interpreted as
follows. 
Given a bipartite graph $G=(U \dcup W, E)$, what is the largest $k$ 
such that the machine node set $W$ is a feasible solution for
\vrpm\ when the scenarios are defined by $\uset_k = \{ F \subseteq W
\colon |F| = k \}$. 
To answer the question, the authors exploit in
\cite{laroche_et_al_14} the so-called $k$-extended Hall's condition: 
\begin{equation}
\label{eq:extended_halls_condition}
\tag{k-Hall}
\forall ~\emptyset \neq T \subseteq U : |T| + k \leq |\neighb_G(T)|
\end{equation}
which guarantees the existence of a node-robust assignment w.r.t.\ the
uncertainty set $\uset_k$. 
Because $G$ is bipartite, this question can be answered efficiently by solving an \ilp\ over an integral polytope $|U|$ times.

\begin{remark*}
  The results from Laroche et al.~\cite[Cor.~2]{laroche_et_al_14} imply
  that feasibility testing for \vrpm\  can be performed in polynomial time when the list 
  $\uset$ of scenarios is given implicitly.
  This is a major difference to \rap\ where deciding upon feasibility of
  an instance with implicitly given scenarios is an \np-complete problem
  (cf. \cite{icalp_full_version_16}), 
  which is an immediate consequence of the \np-completeness of the
  so-called Matching Preclusion Number Problem
  (Lacroix et al.~\cite[Thm.~6]{lacroix_et_al_12}, Dourado et al.~\cite[Thm.~2]{dourado_et_al_15})
\end{remark*}

Zenklusen~\cite{zenklusen2010matching} considered the Matching Interdiction
problem that asks for a given node-weighted graph and a budget to find a subset
of nodes that respects the budget constraint and that minimizes the
size of a maximum matching when the selected nodes are removed from
the graph. 

Arulselvan et al.~\cite{arulselvan_et_al_15} analyzed the following variant of assignment problems.
The input consists of an edge-weighted bipartite graph $G = (U \dcup W, E)$ and
lower and upper quotas $l, u \in \Z_{\geq 0}^W$.
The goal is to find a maximum-weight edge set $M$ such that each node in $U$ is 
incident to at most one edge in $M$.
Furthermore, the nodes in $W$ have either to respect the bounds given by the quota functions or not to be used at all.
Their main results yield a classification of several variants of the problem in terms of their complexity.

Katriel et al.~\cite{katriel_et_08} studied a two-stage stochastic
optimization problem on a bipartite graph $G:=(U\dcup W, E)$ with a cost
function on the edge set, that resembles \vrpm. 
The overall optimization task is to compute an edge subset that contains a maximum matching.
In the first stage there is no uncertainty and one can already select some edges at nominal costs.
In the second stage, uncertainty comes into
play in two variants: either the costs on the edges are uncertain, or some of the
 nodes from $W$ are deactivated.  
For both variants the goal in the second stage is to buy additional edges such that the edges bought in the two stages contain a maximum matching at minimum expected costs.
The main results include the derivation of lower bounds on approximation 
guarantees as well as approximation algorithms. 
Furthermore, the authors provide a randomized approximation algorithm
for the robust variant of the stochastic optimization problem under
consideration. 

The connection between the matching
number of a graph and node removal is also studied from a graph theoretical point of view.
For instance, Aldred et al.~\cite{aldred_et_al_07} provided conditions
under which special graph classes (grid graphs and $k$-fold product graphs) remain perfectly matchable 
after node deletions.
Favaron~\cite{favaron_96} investigated and characterized the class of
so-called $k$-factor-critical graphs, i.e.\ graphs on $n$ nodes such
that every subgraph on $n-k$ nodes is perfectly matchable. 
Note that such graphs can not be bipartite, hence those insights can not be applied here.

\section{Hardness and approximability of V-RAP}
\label{sec:hardness-of-V-RAP}

In this section we show that \vrpm\ is hard to approximate within
a factor of $d \log n$ for any $d < 1$ and provide
an approximation algorithm with a matching approximation ratio (up to an additive constant) \\   
The hardness of approximation is derived in
Section~\ref{sec:reductionfromsetcover}. As in the case for the 
Edge-Robust Assignment Problem
\cite{adjiashvili_bindewald_michaels_icalp2016,icalp_full_version_16}, 
the key ingredient is a reduction from the set cover
problem. \\
The approximation algorithm for \vrpm\ is presented and discussed in
Section~\ref{sec:log-n-approximation-for-V-RAP}. The core idea of
the algorithm is the following. The algorithm first determines a certain matching and
selects all machine nodes covered by the matching. In order to extend
this set of machines to a node-robust solution, a set cover instance
is constructed on all job nodes not saturated by the machine nodes
selected through the matching. The set cover instance is then solved
with the greedy algorithm~\cite{chvatal_79}.
 
\subsection{Hardness for \vrpm} 
\label{sec:reductionfromsetcover}
 
To prove our hardness result for \vrpm, we present a reduction
from the \np-hard set cover problem.  
\begin{problem*}[{Set Cover Problem (\setcover)}] 
  \label{prob:set-cover}
  \text{}
  \begin{itemize}
  \item \underline{Input:} Tuple $([k], \mS)$ with a finite ground set
    $[k] = \{1,\dots, k\}$ and a collection $\mS:=\{S_1, \dots, S_l\}$ of subsets of $[k]$, for some $k,l\in\Z_{\geq 1}$.    
  \item 
    \underline{Output:} Collection $\mC \subseteq \{S_1, \dots,
    S_l\}$ with $\bigcup_{S \in \mC} S =
    [k]$ minimizing $|\mC|$.
  \end{itemize}  
\end{problem*}
For a given instance $([k],\mS)$ of \setcover, existence of any cover for ground set
$[k]$ can be efficiently verified, simply by checking if $\bigcup_{S
  \in \mS} S = [k]$ holds.  
Thus, we will assume from now on that any \setcover\ instance considered is feasible.

Now let $\inst:=([k],\mS)$ be any feasible instance of
\setcover. We associate with $\inst$ the following instance
$\inst^\prime:=(G,\uset,c)$ of \vrpm.
The graph $G$
is obtained by applying the following steps.
\begin{itemize}
  \label{tsteps-v-rap-one-to-two}
\item[\textbf{\small (T1)}]
  For each $s\in [k]$, node $u_s$ is introduced and added to
  $U_{[k]}$. 
  For each $S_j\in \mS$, node $w_{S_j}$ is introduced and added to
  $W_{\mS}$.
  Furthermore, the edge $\{u_s, w_{S_j}\}$ is introduced and added to 
  $E_{\textnormal{\scriptsize SC}}$ whenever $s\in S_j$. 
\item[\textbf{\small (T2)}]
  For each $s\in [k]$, a copy $w_s$ of $u_s$ is introduced
  and added to $W_{[k]}$, and the edge $\{u_s,w_s\}$ is added
  to $E_{[k]}$. 	  
\end{itemize}
Note that edges from $E_{\textnormal{\scriptsize SC}}$ encode whether
an element $s\in [k]$ is contained in a subset $S\in \mS$, or not. 
Nodes from $W_{[k]}$ are used to ensure the feasibility of the \vrpm\
instance, while the nodes from $W_{\mS}$ indicate which elements from
$\mS$ are chosen to cover the ground set $[k]$.

Applying steps \textbf{\small (T1)} and \textbf{\small (T2)} yields the
graph $G:=(U_{[k]}\dcup (W_{\mS}\cup W_{[k]}),\,
E_{\textnormal{\scriptsize SC}}\cup E_{[k]})$. An example of such a graph is illustrated in
Figure~\ref{fig:v-rap-graph-corresponding-to-a-weighted-instance-of-set-cover}. 
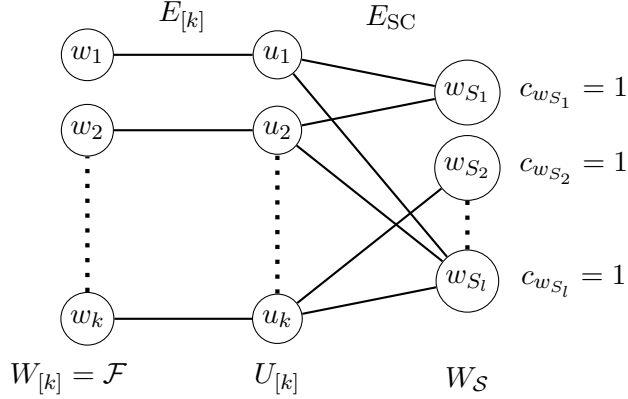
\begin{figure}[htb]
  \centering
  \begin{tikzpicture}[inner sep = 2pt,every node/.style={circle, draw}]
    \node (0) at (-1.5, 0) [] {$w_k$};
    \node (1) at (-1.5, 2.5) [] {$w_2$};
    \node (2) at (-1.5, 3.5) [] {$w_1$};
    
    \node (3) at (1.0, 0) [] {$u_k$};
    \node (4) at (1.0, 2.5) [] {$u_2$};
    \node (5) at (1.0, 3.5) [] {$u_1$};
    
    \node (6) at (3.5, 0.5) [] {$w_{S_l}$};
    \node (7) at (3.5, 2) [] {$w_{S_2}$};
    \node (8) at (3.5, 3) [] {$w_{S_1}$};
    
    \node at (-1.75,-0.8) [draw=none] {$W_{[k]}=\uset$};
    \node at (-0.25,4) [draw=none] {$E_{[k]}$};
    \node at (1.0,-0.8)  [draw=none] {$U_{[k]}$};
    \node at (2.5,4) [draw=none] {$E_{\textnormal{\scriptsize SC}}$};
    \node at (3.5,-0.8)  [draw=none] {$W_{\mS}$};
    
    \node at (4.9,0.5)  [draw=none] {$c_{w_{S_l}} = 1$};
    \node at (4.9,2)  [draw=none] {$c_{w_{S_2}} = 1$};
    \node at (4.9,3)  [draw=none] {$c_{w_{S_1}} = 1$};
    
    \draw[-, thick] (3) -- (7);
    \draw[-, thick] (3) -- (6);
    \draw[-, thick] (4) -- (8);
    \draw[-, thick] (4) -- (6);
    \draw[-, thick] (5) -- (6);
    \draw[-, thick] (5) -- (8);
    
    \foreach \i in {0,..., 2}
    {
      \pgfmathtruncatemacro{\head}{\i}
      \pgfmathtruncatemacro{\tail}{\head + 3}
      \draw (\head) edge [-,thick] (\tail);
    }
    
    \foreach \i in {0,..., 2}
    {
      \pgfmathtruncatemacro{\head}{\i * 3}
      \pgfmathtruncatemacro{\tail}{\head + 1}
      \draw[-, loosely dotted, ultra thick] (\head) -- (\tail);	
    }
    
  \end{tikzpicture}
  \caption{
    Graph $G$ corresponding to a Set Cover instance 
    with ground set $[k]$ and covering sets $S_1,\dots, S_l$.
  }
  \label{fig:v-rap-graph-corresponding-to-a-weighted-instance-of-set-cover}
\end{figure} 

To complete the construction of $\inst^\prime$ we set $\uset = W_{[k]}$ and  
\begin{align}
  \label{eq_redcostfunc}
  c\in\R^{W_{\mS}\cup W_{[k]}}_{\geq 0}
  \quad\textnormal{with}\quad
  w\mapsto c_w:=
  \left\{
  \begin{array}{ll}
    1, & \textnormal{ if } w\in W_{\mS},\\
    0, & \textnormal{ if }  w\in W_{[k]}.
  \end{array}
         \right.
\end{align}
The latter transformation can clearly be carried out in polynomial time.

The next lemma highlights the relation between the \setcover\ and \vrpm instances. 
\begin{lemma}
  \label{lem:set-cover-v-rap-feasibility-version-1}
  Let  $\inst:=([k],\mS)$ be a (feasible) instance of \setcover,
  and let $\inst^\prime:=(G,\uset,c)$ be the corresponding \vrpm\ instance
  with $G:=(U_{[k]}\dcup (W_{\mS}\cup W_{[k]}), E_{[k]}\cup
  E_{\textnormal{\scriptsize SC}})$ obtained by applying steps 
  \textbf{\small (T1) -- (T2)}, uncertainty set $\uset=W_{[k]}$ and cost
  vector $c\in\R^{W_{\mS}\cup W_{{[k]}}}_{\geq 0}$ as specified in
  Equation~\eqref{eq_redcostfunc}. Then, 
    for $\rset \subseteq W_{\mS}\cup W_{[k]}$ with $W_{[k]} \subseteq \rset$, it holds that
    $\rset$ is feasible to $\inst^\prime$ if and only if
    $\mathcal{C}_{\rset} := \{S_j \in \mS \,\mid\, w_{S_j} \in \rset \}$ is
    feasible to $\inst$. Furthermore, such sets $\rset$ and $\mathcal{C}_{\rset}$ have 
    identical costs.
\end{lemma}
\begin{proof}
  Let $\inst:=([k],\mS)$ be a given \setcover\ instance, and let
  $\inst^\prime:=(G,W_{[k]},c)$ be the corresponding \vrpm\ instance.
    \emph{``only if'' part.\quad}
    Let $\rset$  be any feasible solution to $\inst^\prime$ with
    $W_{[k]} \subseteq\rset$, and let $s\in [k]$.  
    To show that $s$ is contained in some set of $\mathcal{C}_{\rset}$, consider
    the node $w_s \in\uset$. As $X$ is feasible to $\inst^\prime$,
    there must exist a $U_{[k]}$-perfect matching $M$ in $G[U_{[k]} \dcup \rset] -
    \{w_s\}$.   
    Since $w_s \notin V(M)$, node $u_s$ must be matched with some node
    from $\{w_{S_1},\dots, w_{S_l}\}$ by the corresponding edge from
    $E_{\textnormal{\scriptsize SC}}$, i.e.\ $w_{S_j} \in \rset$, for some $S_j \in \mS$ with
    $s\in S_j$. This implies that $S_j\in \mathcal{C}_{\rset}$. 
    It follows that $s$ is covered, and that $\mathcal{C}_{\rset}$ is a
    feasible cover for $\inst$. 
    
    \emph{``if'' part.\quad}
    Let $\mC\subseteq \mS$ be a feasible cover for $\inst$. Then,
    define
    \begin{equation*}
      {\rset}:= W_{[k]} \;\cup\;\{w_{S_j}\mid S_j\in\mathcal{C}\},
      \textnormal{ implying that } \mathcal{C}=\mathcal{C}_X \textnormal{ holds}.
    \end{equation*}
    Recall that $X$ is feasible to $\inst^\prime$ if and only
    if $\rset \setminus \{ w_s \} $ contains an assignment of $G$, for
    all $w_s \in\uset=W_{[k]}$. 
    Consider an arbitrary $w_s \in\uset$. 
    We have to show that $G[U_{[k]} \dcup X]-\{w_s\}$ contains a $U_{[k]}$-perfect
    matching. 
    Let $S_j \in \mathcal{C}$ be any set covering $s$. 
    Such a set exists, since $\mathcal{C}$ is a cover. 
    The desired perfect matching can now be defined as
    \begin{align*}
      M := \{u_s, w_{S_j}\} \cup E_{[k]}\setminus \big\{\{w_s, u_s \}\big\}. 
    \end{align*}

    Finally, the costs of the two solutions is clearly identical by the definition
    of the reduction.
\end{proof}

\begin{theorem}
  \label{thm:hardness-of-V-RAP}
  Unless \np\ $\subseteq \dtime (n^{\log \log n})$,
  \vrpm\ admits no polynomial $d\log n$-approximation algorithm for any $d < 1$, where $n$ represents the number of nodes in the underlying graph.   
\end{theorem}
\begin{proof}
  \citet{FeigeSetCover} showed that, for any $d < 1$,  \setcover\ admits no polynomial time  
  $d\log n$-approximation algorithm ($n$ being the size of the ground set) 
  unless we have that \np$\,\subseteq\, \dtime (n^{\log \log n}$). Combining this result with 
  Lemma~\ref{lem:set-cover-v-rap-feasibility-version-1} yields the proof.
\end{proof}

We conclude by showing that the result in Theorem~\ref{thm:hardness-of-V-RAP} also 
holds the more restricted uniform case, in which every machine is vulnerable.

\begin{proposition}
  \label{prop_hardness-of-uniform-V-RAP}
  Theorem~\ref{thm:hardness-of-V-RAP} even holds for the uniform case, i.e.\ with $\uset = W_{\mS}\cup W_{[k]}$.
\end{proposition}
\begin{proof}
  The proof follows from the same reduction with the only difference
  that the set of vulnerable machine nodes contains all machine nodes (resulting
  in a uniform instance). The fact that all machines in $W_{[k]}$ have cost zero implies
  that we can assume them to be part of any solution. The inclusion of these machines 
  allows to match all jobs in the case that any machine in $W_{\mS}$ fails. The rest
  of the proof remains the same. 
\end{proof} 

\subsection{A $(\log |U| +2)$-approximation for \vrpm}
\label{sec:log-n-approximation-for-V-RAP}

In this section we present a $(\log|U| + 2)$-approximation algorithm
for \vrpm\ stated formally as Algorithm~\ref{alg:v-rap-log-n-plus-one-approximation}.
For a node $u\in U \dcup W$ we denote by $\delta(v) \subseteq E$ the set of edges 
incident to $u$.

The algorithm starts by computing a $U$-perfect matching $M$
(step~\ref{alg:step:v-rap-log-n-approximation-computing-matching})
minimizing the cost of job nodes $W^M$ that are incident to the matching.
This is easily achieved by defining an appropriate cost function
on the edge set and computing a minimum cost perfect matching.
The set $W^M$ is chosen to be part of the solution
(step~\ref{alg:step:v-rap-log-n-approximation-initial-machine-set})
and then extended to a feasible \vrpm \ solution as follows.  
The algorithm identifies all job nodes that are matched by $M$ with a vulnerable  
machine node in $W^M$, yielding the set $U_\uset$
(step~\ref{alg:step:v-rap-log-n-approximation-define-critical
  nodes}) called \emph{critical nodes}. 
We think of the matching $M$ as the basis of any $U$-perfect matching for 
each of the failure scenarios, and the additional machine nodes we 
add are designed to replace some vulnerable machine in $W^M$, in case it fails.
The set $U_\uset$ is hence the set of jobs that may become unmatched 
by removing from $M$ edges that are incident to vulnerable nodes.

When a node in $u\in U_\uset$ becomes unmatched in $M$ due to failure of its
corresponding machine $w\in W^M$, it is possible to obtain a new $U$-perfect 
matching by finding an $M$-alternating path starting at $u$, ending with
some node $w'\in W\setminus W^M$ and not using the edge $\{u,w\}$. Such a path
starts and ends with an edge not in $M$, and hence it can be used to increase
the matching size by one, resulting in a new $U$-perfect matching. To allow
this new matching we only need to include the end-node $w'$ of this path
in the solution. We can hence think of $w'$ as a node in $W\setminus W^M$
\emph{covering} the scenario corresponding to $u$. 

This covering interpretation is then used in the algorithm to define a set 
covering problem.
For each machine node $w\in W\setminus W^M$, the algorithm
determines a subset $R_w \subseteq U_\uset$ of critical job nodes $u$ (steps
\ref{alg:step:v-rap-log-n-approximation-compute-cover-subsets-begin} 
-- \ref{alg:step:v-rap-log-n-approximation-compute-cover-subsets-end})
with the property that there is an $M$-alternating $u$-$w$ path in $G$.
Finally, a weighted set cover instance $\inst^{\textnormal{\scriptsize
    SC}}$ with ground set $U_\uset$ and
with the collection of subsets $R_w$, $w\in W\setminus W^M$, is
constructed
(step~\ref{alg:step:v-rap-log-n-approximation-construct-set-cover-instance})
and approximately solved by the greedy algorithm
(step~\ref{alg:step:v-rap-log-n-approximation-set-cover-approximation}).   

\begin{algorithm}[htb]
	\caption{: A $(\log|U| + 2)$-approximation for \vrpm}
	\begin{algorithmic}[1]	
		\label{alg:v-rap-log-n-plus-one-approximation}
		\REQUIRE{ A feasible \vrpm-instance $\inst$: $G=(U \dcup W, E)$, $c \in \R^W_{> 0}$, $\uset \subseteq W$.}
		\ENSURE{ A feasible solution $\rset$. }
		\STATE{ Define an auxiliary cost function $d \in \R^E$: for each
			node $w \in W$ and for each $e \in \delta(w)$ set $d_e := c_w$.} 
		\STATE{ $ M \gets $ minimum-cost $U$-perfect matching w.r.t. cost function $d$ }
		\label{alg:step:v-rap-log-n-approximation-computing-matching}
		\STATE{ $W^M \gets V(M) \cap W$ }
		\label{alg:step:v-rap-log-n-approximation-initial-machine-set}
		\STATE{ $U_\uset \gets \{ u \in U \mid u$ is matched to a
			vulnerable machine node in $M \}$ }
		\label{alg:step:v-rap-log-n-approximation-define-critical nodes}
		\FOR{each machine node $w \in W \setminus W^M$ }
		\label{alg:step:v-rap-log-n-approximation-compute-cover-subsets-begin}
		\STATE{%
			$R_w \gets 
			\{ u \in U_\uset \mid G[U \dcup  (W^M\cup\{w\})]$ contains an
			$M$-alternating $u$-$w$-path$\}$
		}
		\label{alg:step:v-rap-log-n-approximation-aternating-paths}
		\ENDFOR 
		\label{alg:step:v-rap-log-n-approximation-compute-cover-subsets-end}
		\STATE{Construct a set cover instance
			$\inst^{\textnormal{\scriptsize SC}}:= \big(U_\uset,\,\{ R_w
			\mid w \in W \setminus W^M \}\big)$ with weight function $g: \{ R_w
			\mid w \in W \setminus W^M \}\rightarrow\R_{\geq 0}$,
			$R_w\mapsto g(R_w):= c_w$}
		\label{alg:step:v-rap-log-n-approximation-construct-set-cover-instance}
		\STATE{Apply the greedy algorithm to the set cover instance 
			$\inst^{\textnormal{\scriptsize SC}}$ to obtain an approximate
			solution $W^{\textnormal{\scriptsize SC}}$.}
		\label{alg:step:v-rap-log-n-approximation-set-cover-approximation}
		\RETURN{$\rset = W^M \dcup W^{\textnormal{\scriptsize SC}}$}
		\label{alg:step:v-rap-log-n-approximation-output-X}
	\end{algorithmic} 
\end{algorithm}

As we informally sketched above, the approximate solution of 
$\inst^{\textnormal{\scriptsize SC}}$ defines a set $W^{\textnormal{\scriptsize SC}}$
of machine nodes whose addition to $W^M$ results in a feasible 
solution $\rset$ for the \vrpm\ instance.
This property is established by the next lemma. 
\begin{lemma}  
  \label{lem:feasible-solution-imply-solutions-to-set-cover}
  Let $\inst= (G,\uset,c)$ be a feasible \vrpm\
  instance with $G:=(U\dcup W, E)$, and let $M$ be any $U$-perfect matching in $G$. We denote by
  $W^M$ the machine nodes covered by $M$. 
  Furthermore, let $U_\uset$ be the set of all critical job nodes as determined in
  step~\ref{alg:step:v-rap-log-n-approximation-define-critical nodes}
  of Algorithm~\ref{alg:v-rap-log-n-plus-one-approximation}, and
  let $X\subseteq W$ be any subset of machine nodes containing $W^M$. 
  Define the set system $\mathcal{R}^{\textnormal{\scriptsize SC}}_X:=\{R_w\mid w\in
  X\setminus W^M\}$ where $R_w$ represents the set of
  critical job nodes calculated in
  step~\ref{alg:step:v-rap-log-n-approximation-aternating-paths} of
  Algorithm~\ref{alg:v-rap-log-n-plus-one-approximation}. \\
  Then, $X$ is feasible for $\inst$ if and only if
  $\mathcal{R}^{\textnormal{\scriptsize SC}}_X$ forms a cover for
  $U_\uset$ (i.e.\ $\mathcal{R}^{\textnormal{\scriptsize SC}}_X$ is a
  feasible solution for the set cover instance
  $\inst^{\textnormal{\scriptsize SC}}$ constructed in  
  step~\ref{alg:step:v-rap-log-n-approximation-construct-set-cover-instance}
  of Algorithm~\ref{alg:v-rap-log-n-plus-one-approximation}).  
\end{lemma}
\begin{proof}
  Let $X\subseteq W$ with $W^M\subseteq X$ be a feasible solution to
  $\inst$, and consider an arbitrary node $u \in U_\uset$.
  We have to show that there exists a set $R_w$ containing $u$, for
  some $w \in X \setminus W^M$. 
  \\
  Now, let $f$ be the vulnerable machine matched to $u$ by $M$.
  We consider the matching  $\hat M:= M \setminus \{\{u,f\}\}$ of size
  $|U| - 1$.
  As $\rset$ is feasible, the subgraph $H := G[U \dcup X] - \{f\}$
  contains a $U$-perfect matching that has size $|U|$. 
  Since $\hat M$ is also a matching in $H$, it follows from a result
  by Berge that it can be augmented to a $U$-perfect matching using an
  $\hat M$-augmenting path $\hat P$. 
  The path $\hat P$ starts in $u$ and ends in some node $w \in X \setminus W^M$.
  Since $\hat P$ is $\hat M$-augmenting, $w$ is the only node in
  $\hat P$ that is not incident to $\hat M$, i.e.\ $\hat P$ is a path in 
  $G[U\dcup (W^M \cup \{w\})]$. 
  By construction of $\hat M$, $\hat P$ must be an odd $M$-alternating path
  and hence $u \in R_w$.
  As $u\in U_\uset$ was chosen arbitrarily, we can conclude that
  $\mathcal{R}_X^{\textnormal{\scriptsize SC}}$ forms a cover of
  $U_{\uset}$ and is hence feasible to $\inst^{\textnormal{\scriptsize SC}}$. 

  For the reverse direction, let $\mathcal{R}^{\textnormal{\scriptsize SC}}\subseteq
  \{R_w\mid w\in W\setminus W^M\}$ be any feasible solution to the set
  cover instance $\inst^{\textnormal{\scriptsize SC}}$ as constructed in
  step~\ref{alg:step:v-rap-log-n-approximation-construct-set-cover-instance}
  of Algorithm~\ref{alg:v-rap-log-n-plus-one-approximation}.
  We have to show that 
  $\rset=W^M\dcup \{w\in W\setminus 
  W^M\mid R_w\in \mathcal{R}^{\textnormal{\scriptsize SC}}\}$ is
  feasible for $\inst$.
  For this, we prove that for each vulnerable machine node $f\in \uset$,
  there exists a $U$-perfect matching not using $f$ in the
  graph $G[U \dcup \rset]$.  
  Note that $M$ is contained in $G[U \dcup \rset]$,
  i.e.\ $M$ provides the desired $U$-perfect matching for all $f \in
  \uset$ not incident to $M$.\\
  Now, consider an arbitrary $f \in \uset$ incident to $M$. 
  We denote by $u$ the job node matched to $f$ in $M$. 
  By definition, $u$ is a critical job node, i.e.\ $u\in U_\uset$.
  From the fact that $\mathcal{R}^{\textnormal{\scriptsize SC}}$ forms
  a cover of  $U_\uset$,  
  it follows that $u \in R_{w}$, for some $w \in X \setminus W^M$.
  By definition of $R_w$,  there exists an $M$-alternating path $P$
  from $u$ to $w$ in $G[U\dcup (W^M\cup \{w \})]$. 
  As $w$ is not covered by $M$, the path $P$ ends with an edge (incident to
  $w$) that does not belong to $M$. Note that $P$ has an odd number of edges. Thus, the first edge of $P$
  incident with $u$ is also not contained in $M$, i.e.\ $\{u,f\}\notin
  P$. 
  Thus,  $M \bigtriangleup (P + \{u,f\})$ is a $U$-perfect matching in
  $G[U \dcup \rset]$ not using $f$.  
\end{proof}    
   
We are now ready to prove the main result of this section.
\begin{theorem}
  \label{thm:v-rap-log-n-plus-one-algorihm}
  Algorithm~\ref{alg:v-rap-log-n-plus-one-approximation} 
  is a polynomial $(\log |U| + 2)-$approximation algorithm for \vrpm.
\end{theorem}
\begin{proof}
  The algorithm can clearly be implemented in polynomial time.
  For a reference on algorithms for minimum cost perfect matching 
  computations and augmenting path computations we refer the reader to~\cite{lovasz_plummer_book_86}.

  Let $\mathcal{R}^{\textnormal{\scriptsize SC}}$ be the
  approximate solution of the set cover instance computed by the
  greedy algorithm.
  Then, Algorithm~\ref{alg:v-rap-log-n-plus-one-approximation} returns
  $X=W^M\dcup W^{\textnormal{\scriptsize SC}}$ where  $W^{\textnormal{\scriptsize SC}}:= \{w\in W\setminus 
  W^M\mid R_w\in \mathcal{R}^{\textnormal{\scriptsize SC}}\}$.
  As $\mathcal{R}^{\textnormal{\scriptsize SC}}$ forms a cover of
  $U_\uset$, it follows from
  Lemma~\ref{lem:feasible-solution-imply-solutions-to-set-cover} that
  $X$ is feasible for $\inst$. 

  It remains to prove that the computed solution $\rset$ satisfies the
  desired quality. For this, let $\opt$ and $\opt^{\textnormal{\scriptsize SC}}$ be optimal
  solutions for the given \vrpm\ instance $\inst$ and for the
  constructed set cover instance $\inst^{\textnormal{\scriptsize SC}}$ with associated cost $c(\opt)$ and
  $
  g(\opt^{\textnormal{\scriptsize SC}}) = \sum_{R_w\in
    \textnormal{\scriptsize OPT}^{\textnormal{\footnotesize SC}}} g(R_w) = \sum_{R_w\in
    \textnormal{\scriptsize OPT}^{\textnormal{\scriptsize SC}}}  c_w.  
  $
  The cost of the returned solution $X$ is
  \begin{align}
    \label{eq_thm:v-rap-log-n-plus-one-algorihm_cost_X}
    c(X) = c(W^M) + c(W^{\textnormal{\scriptsize SC}}).
  \end{align}
  As $M$ is chosen in
  step~\ref{alg:step:v-rap-log-n-approximation-computing-matching} to
  be a $U$-perfect matching minimizing the cost of the incident machine 
  nodes and since any feasible \vrpm\ solution contains some $U$-perfect matching we obtain
  \begin{align}
    \label{eq_thm:v-rap-log-n-plus-one-algorihm_cost_WM}
    c(W^M)\leq c(\opt).
  \end{align}
  Since we solve the set cover instance with the greedy algorithm
  we have
  \begin{align}
    \label{eq_thm:v-rap-log-n-plus-one-algorihm_approx-cost-estimate}
    c(\mathcal{R}^{\textnormal{\scriptsize SC}}) \leq (\log |U_\uset| + 1) \cdot
    c(\textnormal{\opt}^{\textnormal{\scriptsize SC}}),
  \end{align}
  due to~\cite{chvatal_79}.
  As $\opt$ is feasible to the \vrpm\ instance
  $\inst$,
  Lemma~\ref{lem:feasible-solution-imply-solutions-to-set-cover}
  implies that the associated collection $\mathcal{R}^{\textnormal{\scriptsize
      SC}}_{\textnormal{\scriptsize OPT}}=\{R_w \mid w\in \opt\setminus\, W^M\}$ is feasible to the set cover 
  instance $\inst^{\textnormal{\scriptsize SC}}$. Thus, we further have that 
  \begin{align}
    \label{eq_thm:v-rap-log-n-plus-one-algorihm_costestimate-optsc}
    g\left(\opt^{\textnormal{\scriptsize SC}}\right)\;\leq\; g\left(\mathcal{R}^{\textnormal{\scriptsize
    SC}}_{\textnormal{\scriptsize OPT}}\right) = \sum_{R_w\in
    {{\mathcal{R}}}^{\textnormal{\tiny
    SC}}_{\textnormal{\tiny OPT}}} c_w = c(\opt\setminus\,W^M) 
    \leq c(\opt).
  \end{align}
  Combining the results from
  Equations~\eqref{eq_thm:v-rap-log-n-plus-one-algorihm_cost_X}--\eqref{eq_thm:v-rap-log-n-plus-one-algorihm_costestimate-optsc}
  and using $U_{\uset} \subseteq U$, we obtain the desired bound on the approximation guarantee
  \[ c(\rset) \leq (\log |U| +2) \cdot c(\opt). \]
\end{proof}

We note that the approximation guarantee we obtain is actually 
$\log |U_\uset| +2$, which is at most
$$
\log \min\{ |U|, |\uset|\} + 2,
$$
since $|U_\uset| \leq \min\{ |U|, |\uset|\}$ by definition of the matching $M$ and $U_\uset$.
We hence obtain a better guarantee if the number of faulty machine
nodes is significantly smaller than $|U|$.

\section{Hardness and approximability of card-V-RAP}
\label{sec:hardness-of-card-V-RAP}

This section deals with the unweighted version \uvrpm. 
Similarly to the unweighted
edge-robust variant \urap, we first prove in
Section~\ref{sec:hardnessofcardvrap} that \uvrpm\ does not admit a
\ptas, provided that $\p\neq\np$. This is achieved by
reducing the vertex cover problem in sub-cubic graphs to \uvrpm\ and
invoking a result of Alimonti and Kann \cite{alimonti_kann_97} for
the former problem.
In Section~\ref{sec:card-v-rap-admits-a-o(1)-approximation}, we
refine the analysis of our approximation algorithm presented for
general \vrpm.
In particular, we show that this algorithm is a constant
factor approximation algorithm for \uvrpm\ instances.
In Section~\ref{sec:cardvrap-twovulnerablemachinecase} we
consider the case when the number of vulnerable machine nodes is two. 
We prove that \uvrpm\ is solvable in polynomial 
time in this case.

\subsection{Hardness of card-V-RAP}
\label{sec:hardnessofcardvrap}

In this section we prove the following result for \uvrpm. 
\begin{theorem} 
  \label{thm:card-v-rap-is-hard-to-approximate}
  For some constant $\delta > 1$, there is no polynomial $\delta$-approximation
  for uniform \uvrpm, unless $\p = \np$.
\end{theorem}
To prove Theorem~\ref{thm:card-v-rap-is-hard-to-approximate}, we
resort to the reduction from Set Cover discussed in
Section~\ref{sec:hardness-of-V-RAP}. Since we now wish to deal with
uniform weights on the machine nodes, the former reduction needs to
be adjusted.
In particular, we cannot assume that $W_{[k]}$ is contained in some
optimal solution since including all of $W_{[k]}$ might always be sub-optimal. 
To be more precise, let $\inst:=([k],\mS)$ be a feasible instance of
\setcover, and let $G:=(U_{[k]}\dcup (W_{\mS}\cup W_{[k]}), E_{\textnormal{\scriptsize SC}}\cup
E_{[k]})$ be the resulting graph obtained by performing steps {\small\bf
  (T1)} and  {\small\bf (T2)} from
Section~\ref{sec:hardness-of-V-RAP}. We apply the following additional step to $G$. 
\begin{itemize}
  \label{v-rap-tstep-three}
\item[{\small\bf (T3)}] 
  For each $s\in [k]$, two further copies $\bar u_s$ and $\bar w_s$ are
  introduced and added to $\bar U_{[k]}$ and to $\bar W_{[k]}$, respectively.
  Then, for each $s\in[k]$, the edges $\{w_s,\bar u_s\}$ and $\{\bar
  u_s,\bar w_s\}$ are introduced and added to $E_{[k]}$.
\end{itemize}
Let $\bar G:=(\bar U \dcup \bar W, \bar E)$ be the resulting graph with
$\bar U:= U_{[k]}\cup \bar U_{[k]}$, 
$\bar W:= W_{\mS}\cup W_{[k]}\cup \bar W_{[k]}$
and
$\bar E:= E_{\textnormal{\scriptsize SC}}\cup E_{[k]}$.
An example of such a graph is illustrated in 
Figure~\ref{fig:v-rap-graph-corresponding-to-an-unweighted-instance-of-minimum-cover}. 
Note that the size of $\bar G$ is still polynomial in $k$ and $l$, the
input length of the corresponding set cover instance $\inst$.
We can now prove Theorem~\ref{thm:card-v-rap-is-hard-to-approximate}.
\begin{figure}[htb]
	\centering
	\begin{tikzpicture}[inner sep = 2pt,every node/.style={circle, draw}]
	\node (0) at (-3, 0) [] {$\bar{w}_k$};
	\node (1) at (-3, 2.5) [] {$\bar{w}_2$};
	\node (2) at (-3, 3.5) [] {$\bar{w}_1$};
	
	\node (3) at (1.5, 0) [] {$u_k$};
	\node (4) at (1.5, 2.5) [] {$u_2$};
	\node (5) at (1.5, 3.5) [] {$u_1$};
	
	\node (6) at (3.5, 0.5) [] {$w_{S_l}$};
	\node (7) at (3.5, 2) [] {$w_{S_2}$};
	\node (8) at (3.5, 3) [] {$w_{S_1}$};
	
	\node (9) at (-1.5, 0) [] {$\bar{u}_k$};
	\node (10) at (-1.5, 2.5) [] {$\bar{u}_2$};
	\node (11) at (-1.5, 3.5) [] {$\bar{u}_1$};
	
	\node (12) at (0, 0) [] {${w}_k$};
	\node (13) at (0, 2.5) [] {${w}_2$};
	\node (14) at (0, 3.5) [] {${w}_1$};
	
	\node at (3.5,-0.8) [draw=none] {$W_{\mS}$};
	\node at (1.5,-0.8) [draw=none] {$U_{[k]}$};
	\node at (0.0,-0.8) [draw=none] {$W_{[k]}$};
	\node at (-1.5,-0.8) [draw=none] {$\bar{U}_{[k]}$};
	\node at (-3,-0.8) [draw=none] {$\bar{W}_{[k]}$};
	
	\draw[-, thick] (3) -- (7);
	\draw[-, thick] (3) -- (6);
	\draw[-, thick] (4) -- (8);
	\draw[-, thick] (4) -- (6);
	\draw[-, thick] (5) -- (6);
	\draw[-, thick] (5) -- (8);
	
	\foreach \i in {0,..., 2}
	{
		\pgfmathtruncatemacro{\head}{\i}
		\pgfmathtruncatemacro{\tail}{\head + 9}
		\draw (\head) edge [-,thick] (\tail);
	}
	
	\foreach \i in {9,..., 11}
	{
		\pgfmathtruncatemacro{\head}{\i}
		\pgfmathtruncatemacro{\tail}{\head + 3}
		\draw (\head) edge [-,thick] (\tail);
	}
	
	\foreach \i in {12,..., 14}
	{
		\pgfmathtruncatemacro{\head}{\i}
		\pgfmathtruncatemacro{\tail}{\head - 9}
		\draw (\head) edge [-,thick] (\tail);
	}
	
	\foreach \i in {0,..., 4}
	{
		\pgfmathtruncatemacro{\head}{\i * 3}
		\pgfmathtruncatemacro{\tail}{\head + 1}
		\draw[-, loosely dotted, ultra thick] (\head) -- (\tail);	
	}
	
	\end{tikzpicture}
	
	\caption{
		The graph $\bar G$ corresponding to a Set Cover instance 
		with ground set $[k]$ and covering sets $S_1,\dots, S_l$.
	}
	\label{fig:v-rap-graph-corresponding-to-an-unweighted-instance-of-minimum-cover}
\end{figure}
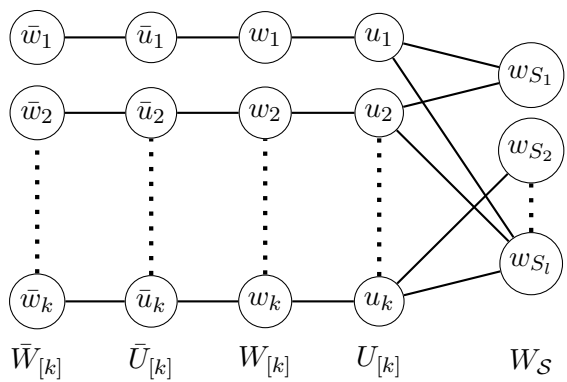

\begin{proof}
  It is well known that the \emph{Vertex Cover Problem in sub-cubic
    Graphs} (\vcthree) on an input graph $H=(V_H,E_H)$ with $|E_H|=k$
  can be equivalently restated as an instance of the Set Cover
  Problem, where the ground set is $E_H\cong[k]$ and each $S_j\in\mS$
  corresponds to a cut set $\delta(v)$, for some $v\in V_H$, i.e.\
  $\mS=\{\delta(v)\mid v\in V_H\}$.  
  As $H$ is sub-cubic, we have that $|S_j|\leq 3$, for all $S_j\in\mS$. Moreover, $|\{S_j\in \mS \,\mid\, s\in S_j\}|
  = 2$ holds for each $s\in [k]$, namely $s\in S_{j_1}$ and $s\in
  S_{j_2}$  where $s$ represents edge $e=\{v_1,v_2\}\in E_H$ and
  $S_{j_1}$ and $S_{j_2}$ correspond to $\delta(v_1)$ and $\delta(v_2)$.
  \\
  Now, let a \vcthree\ instance be presented as a \setcover\ instance
  $([k],\mS)$. Furthermore, let $\bar G$ be the graph obtained from
  $([k], \mS)$ by applying the transformation steps {\small\bf
    (T1)--(T3)}, and let $\bar{\inst}$ be the uniform \uvrpm\ instance
  induced by $\bar G$. 
  \\
  Lemma~\ref{lem:set-cover-v-rap-feasibility-version-1} has to be
  adjusted in the following manner. 
  First note that any feasible solution $\rset \subseteq \bar W$ to
  \vrpm\ with $\uset = \bar W$ has to contain both sets $W_{[k]}$ and
  $\bar W_{[k]}$, because the jobs corresponding to node set  $\bar
  U_{[k]}$ are only adjacent to those nodes. 
  The vulnerability of nodes in $W_{[k]}$ implies again that $\rset$
  corresponds to the feasible cover $\mC_\rset := \{ S_j \mid w_{S_j}
  \in \rset \}$. 
  Conversely, any feasible cover $\mC$ for $([k],\mS)$ gives rise to a
  feasible solution $\rset_\mC := \{ w_{S_j} \mid S_j \in \mC \} \cup W_{[k]}
  \cup \bar W_{[k]}$ for the corresponding uniform \uvrpm\ instance
  with cost $|\rset_\mC | = 2k + |\mC|$.  
  \\
  Next, we observe that $\bar G$ contains $5k$ edges (as $([k],\mS)$
  encodes a \vcthree\ instance) and that any feasible solution of
  \vcthree\ has size of at least $k/3$.     
  Now assume on contrary that for any constant $\delta > 1$, there is
  a polynomial $\delta$-approximation $A_\delta$ for uniform \uvrpm.
  Due to the relation $|\rset_\mC | = 2k + |\mC|$,  $A_\delta$ can be
  used to approximate \vcthree\ within a certain constant factor
  $\alpha(\delta)>1$ where $\alpha(\delta)\rightarrow 1$, for
  $\delta\rightarrow 1$. 
  However, Alimonti and Kann proved in~\cite{alimonti_kann_97} that
  there exists a constant $\alpha >1$  
  such that \vcthree\ does not admit a polynomial $\alpha$-approximation algorithm unless \p$=$\np.
  This completes the proof.
\end{proof}

\subsection{A $O(1)$-approximation algorithm for uniform card-V-RAP}
\label{sec:card-v-rap-admits-a-o(1)-approximation}

In this section we prove the following result.
\begin{theorem}
  \label{thm:card-v-rap-admits-a-1.75-approximation}
  Uniform \uvrpm\ admits a polynomial $1.75$-approximation algorithm.
\end{theorem}
\begin{proof}
We show that Algorithm~\ref{alg:v-rap-log-n-plus-one-approximation}, 
when applied to feasible uniform \uvrpm\ instances, already guarantees
the desired approximation ratio. 
For clarity we rewrite Algorithm~\ref{alg:v-rap-log-n-plus-one-approximation}
in terms of uniform \uvrpm\ to obtain Algorithm~\ref{alg:card-v-rap-1.75-approximation}.
Since we aim to obtain a constant factor approximation ratio we can no longer directly 
rely on the approximation guarantee of the greedy algorithm of an arbitrary set cover 
instance, and instead perform an analysis of the entire algorithm simultaneously. 
We thus write the steps of the greedy algorithm explicitly in 
Algorithm~\ref{alg:card-v-rap-1.75-approximation}.
\begin{algorithm}[H] 
  \caption{: A $1.75$-approximation algorithm for uniform \uvrpm.}
  \begin{algorithmic}[1]
    \label{alg:card-v-rap-1.75-approximation}
    \REQUIRE{ A feasible \uvrpm-instance: $G=(U \dcup W, E)$, $\uset = W$.}
    \ENSURE{ A feasible solution $\rset$.}
    \STATE{ $ M \gets $ any $U$-perfect matching in $G$ } \label{alg:step:card-v-rap-compute-matching}
    \STATE{ $W^M \gets V(M) \cap W$} 
    \label{alg:step:card-v-rap-initial-node-set}
    \FOR{each machine node $w \in W \setminus X$ }
    \STATE{ $R_w \gets \{ u \in U \mid \exists \text{ an } M
      \text{-alternating $u$-$w$-path in } G[U\dcup(W^M\cup\{w\})] \}$ } 
    \ENDFOR
    \STATE{ $\rset \gets W^M$ }
    \STATE{ $U_\uset \gets U$ }
    \WHILE{ $|U_\uset| > 0$ }
    \label{alg:card-v-rap-1.75-approximation-setcoverapproximation-while}
    \STATE{ $\bar w \gets \arg\max \{ |R_w \cap U_\uset| \colon w \in W
      \setminus \rset \} $  }
    \label{alg:card-v-rap-1.75-approximation-approximation-argmaxcomputation}
    \STATE{ $U_\uset \gets U_\uset \setminus R_{\bar w}$ }
    \label{alg:card-v-rap-1.75-approximation-approximation-updatecritcaljobs}
    \STATE{ $\rset \gets \rset \cup \{ \bar w \}$ }
    \label{alg:card-v-rap-1.75-approximation-approximation-updatesolution}
    \ENDWHILE
    \label{alg:card-v-rap-1.75-approximation-setcoverapproximation-endwhile}
    \RETURN{$\rset$}
  \end{algorithmic}
\end{algorithm}

Note that in the uniform \uvrpm\ case, one can start with an arbitrary
$U$-perfect matching $M$ since the cost of any such matching is $|U|$. 
Furthermore, all job nodes are matched to a vulnerable machine node, as
we assume the uniform case. Therefore, each job node is critical and
$U_{\uset}= U$. \\

To prove the quality of the computed solution $\rset$, we proceed as
follows. 
We distinguish two types of iterations of the set cover greedy
subroutine corresponding to
steps~\ref{alg:card-v-rap-1.75-approximation-setcoverapproximation-while}--\ref{alg:card-v-rap-1.75-approximation-setcoverapproximation-endwhile}
of Algorithm~\ref{alg:card-v-rap-1.75-approximation}.  
An iteration is called \emph{productive} if the cardinality of $U_\uset$, 
the set of not yet covered job nodes,
decreases by at least two in
step~\ref{alg:card-v-rap-1.75-approximation-approximation-updatecritcaljobs}
of this iteration. This means that adding the current machine node $\bar w$ computed in
step~\ref{alg:card-v-rap-1.75-approximation-approximation-argmaxcomputation}
to $\rset$ will saturate at least two critical job nodes.   
All other iterations are called \emph{nonproductive}.
\\
Let $p$ be the total decrease of $U_\uset$ obtained from all productive iterations, and let $\opt\subseteq W$ denote an optimal solution for
the given uniform \uvrpm\ instance.
We next prove two claims that we will use to derive the desired
approximation ratio.
\begin{itemize}
\item \textbf{Claim 1: } $|\rset| \leq 2 |U| - \tfrac{p}{2}$
   
  In step~\ref{alg:step:card-v-rap-initial-node-set}, $\rset$ is set
  to $W^M$, i.e.\ $|X|=|W^M|=|U|$.
  In the uniform case, every job node is critical. Thus,
  the ground set of the set cover instance is $U$.
  Since every productive iteration saturates at least two nodes, there can be
  at most $\tfrac{p}{2}$ productive iterations.
  In each such iteration, one further machine node is added to
  $\rset$. 
  All remaining iterations are nonproductive, and there are exactly $|U| - p$ of
  them. In total, we have 
  \begin{align*}
    |\rset|\leq |U| + \frac{p}{2} + (|U| - p) = 2 |U| - \frac{p}{2}. 
  \end{align*}
  
\item \textbf{Claim 2: } $|\opt| \geq \max\{|U|, 2(|U|-p)\}$

  \opt\ contains at least one $U$-perfect matching, i.e.\  $|\opt| \geq
  |U|$.
  
  Recall that there are $|U| - p$ nonproductive iterations, and that
  in each nonproductive iteration only one additional job node is covered. 
  We denote by $U^\prime$ all job nodes from $U$ being 
  covered in a nonproductive iteration of the set cover greedy subroutine. 
  Then, for any pair of distinct nodes $u_1,u_2\in U^\prime$, $u_1\neq
  u_2$, we observe that their neighborhoods in $G$ are disjoint 
  i.e.\ $\neighb_G(u_1) \cap \neighb_G(u_2) = \emptyset$.
  Indeed, if this were not the case, and $u_1$ and $u_2$ would have a common
  neighbor in $W$, it would be possible to cover them in a productive iteration, 
  contradicting $u_1$, $u_2 \in U^\prime$.
  Since we have a uniform instance, any node in $U$ (and hence in $U^\prime$) 
  must have at least two neighbors in any feasible solution, 
  including $\opt$. This gives us the bound $|\opt| \geq 2(|U| - p)$.
\end{itemize}
Finally, we derive an upper bound on $\frac{|\rset|}{|\opt|}$ corresponding 
to the approximation ratio.
Claim 2 allows us to focus on the following two cases.
\begin{itemize}
\item \textbf{Case 1: } $|\opt| \geq |U| \geq 2(|U|-p)$, i.e.\ $2p \geq |U|$
  
  From Claim 1 and $2p \geq |U|$, we obtain that
  \begin{align*}
    |\rset| \leq 2 |U| - \frac{p}{2} \leq 2|U| - \frac{|U|}{4} \leq
    \frac{7}{4}|U|\leq \frac{7}{4} |\opt| . 
  \end{align*}
  
\item \textbf{Case 2: } $|\opt| \geq 2(|U|-p) \geq |U|$, i.e.\ $2p \leq |U|$
  
  From Claim 1 and $2p \leq |U|$, we derive that
  \begin{align*}
    |\rset| \leq 2 |U| - \frac{p}{2}
    \leq \frac{7}{2} |U| - \frac{7}{2}p \leq \frac{7}{2} (|U|-p) \leq \frac{7}{4} |\opt| .
  \end{align*} 
\end{itemize}
In both cases we obtain $\frac{|X|}{|\opt|} \leq \frac{7}{4}$, which concludes the proof.
\end{proof}

For non-uniform instances of \uvrpm\ the latter proof does not apply, 
since $|\opt| \geq 2(|U|-p)$ need not hold in general. However, we can
can still show that a factor $2$ is achieved.

\begin{corollary}
\label{cor:card-v-rap-admits-a-2-approximation}
Algorithm~\ref{alg:card-v-rap-1.75-approximation} has an approximation ratio of 2 for non-uniform instances of \uvrpm.
\end{corollary}
	
\begin{proof}
Replace the bounds in Claim 2 with the bounds $|\opt| \geq |U|$ and $|\rset| \leq 2|U| - \frac{p}{2}$.
Combining these two bounds yields the ratio of $2$.
\end{proof}
 
\subsection{Two vulnerable machines case}
\label{sec:cardvrap-twovulnerablemachinecase}
One interesting fact about the Edge-Robust Assignment Problem is that this problem is
even \np-hard in its simplest non-trivial variant, i.e.\ in case of two
vulnerable edges and unit weights (see
\cite{adjiashvili_bindewald_michaels_icalp2016}, and \cite[Section
5]{icalp_full_version_16} for a proof).  
In contrast, here we show that the node-robust assignment problem remains
tractable in this
setting. 
This result is summarized in the next theorem. 
\begin{theorem}
\label{thm:v-rap-with-two-scenarios-is-in-p}
\uvrpm\ with $|\uset|=2$ is solvable in polynomial time.
\end{theorem}
\begin{proof}
Let $\inst = (G,\uset)$ be a \uvrpm\ instance with $G:=(U \dcup W,E)$
and $\uset=\{w^\prime,w^{\prime\prime}\}\subseteq W$, $w^\prime\neq 
w^{\prime\prime}$.  
Given an optimal solution $X$ to $\inst$, observe first that either both $w^\prime$ and $w^{\prime\prime}$ are
contained in $X$ or none of them. In the latter case,
an optimal solution is given by an $U$-perfect matching in
$G-\{w^\prime,w^{\prime\prime}\}$.  One can use standard bipartite
matching algorithms to verify the existence of such a matching and, if
existent, corresponds to an optimal solution.
\\
In the remaining case $\uset$ is part of any optimal solution.
We introduce a dummy job node $d$ and the edges
$e^\prime:=\{d,w^\prime\}$ and
$e^{\prime\prime}:=\{d,w^{\prime\prime}\}$. 
Additionally we double every edge that is not incident to one of the
vulnerable nodes $w^\prime$ and $w^{\prime\prime}$. 
This gives us a new graph $G^\prime:=(U^\prime \dcup W^\prime, E^\prime)$ with $U^\prime := U
\cup \{d\}$, $W^\prime= W$ and $E^\prime:= E \cup
\{e^\prime,e^{\prime\prime}\} \cup \{\bar e \mid e \in E \setminus
\delta_G(\{w^\prime,w^{\prime\prime}\})\}$
(see Figure~\ref{fig:v-rap-modified-graph-for-two-vulnerable-machines-case},
for an illustration).  
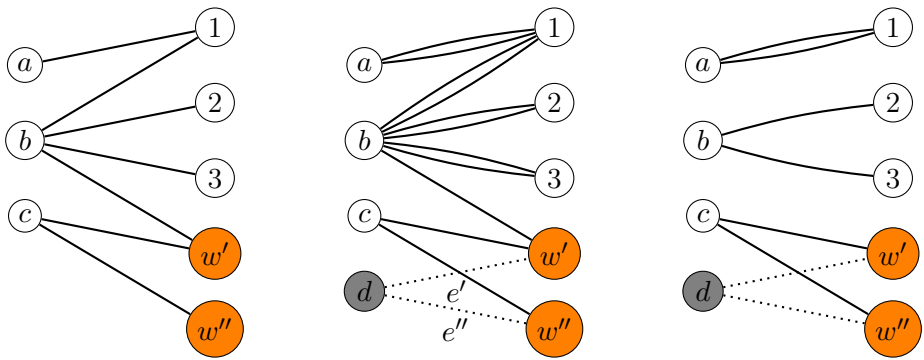
\begin{figure}[htb]
  \centering
  \begin{minipage}{0.2\linewidth}	
    \begin{tikzpicture}[inner sep = 2pt,every node/.style={circle, draw}]
      \def\jx{0}		
      \def\mx{2.5}		
      
      \node (a) at (\jx,1) [] {$a$};
      \node (b) at (\jx,0) [] {$b$};
      \node (c) at (\jx,-1) [] {$c$};
      
      \node (1) at (\mx,1.5) [] {1};
      \node (2) at (\mx,0.5) [] {2};
      \node (3) at (\mx,-0.5) [] {3};
      \node (w1) at (\mx,-1.5) [inner sep = 2pt, fill=orange] {$w^\prime$};
      \node (w2) at (\mx,-2.5) [inner sep = 2pt, fill=orange] {$w^{\prime \prime}$};
      
      \def\angle{0}		
      \path[=,thick, double] 	
      (a)		edge [bend right=\angle] (1)
      (b)		edge [bend right=\angle]  (1)
      edge [bend right=\angle] (2)
      edge [bend right=\angle] (3)
      edge (w1)
      (c)		edge (w1)
      edge (w2);
    \end{tikzpicture}
  \end{minipage}
  \hspace{0.1\linewidth}
  \begin{minipage}{0.2\linewidth}	
    \begin{tikzpicture}[inner sep = 2pt,every node/.style={circle, draw}]
      \def\jx{0}		
      \def\mx{2.5}		
      
      \node (a) at (\jx,1) [] {$a$};
      \node (b) at (\jx,0) [] {$b$};
      \node (c) at (\jx,-1) [] {$c$};
      \node (d) at (\jx,-2) [fill=gray] {$d$};
      
      \node (1) at (\mx,1.5) [] {1};
      \node (2) at (\mx,0.5) [] {2};
      \node (3) at (\mx,-0.5) [] {3};
      \node (w1) at (\mx,-1.5) [inner sep = 2pt, fill=orange] {$w^\prime$};
      \node (w2) at (\mx,-2.5) [inner sep = 2pt, fill=orange] {$w^{\prime \prime}$};
      
      \def\angle{5}		
      \path[=,thick, double] 	
      (a)		edge [bend right=\angle] (1)
      edge [bend left=\angle] (1)
      (b)		edge [bend right=\angle]  (1)
      edge [bend left=\angle] (1)
      edge [bend right=\angle] (2)
      edge [bend left=\angle] (2)
      edge [bend right=\angle] (3)
      edge [bend left=\angle] (3)
      edge (w1)
      (c)		edge (w1)
      edge (w2)
      (d)	    edge [dotted] node[draw=none, below, inner sep = 0.5pt] {$e^\prime$} (w1)
      edge [dotted] node[draw=none, below, inner sep = 0.5pt] {$e^{\prime \prime}$} (w2);
    \end{tikzpicture}
  \end{minipage}
  \hspace{0.1\linewidth}
  \begin{minipage}{0.2\linewidth}	
    \begin{tikzpicture}[inner sep = 2pt,every node/.style={circle, draw}]
      \def\jx{0}		
      \def\mx{2.5}		
      
      \node (a) at (\jx,1) [] {$a$};
      \node (b) at (\jx,0) [] {$b$};
      \node (c) at (\jx,-1) [] {$c$};
      \node (d) at (\jx,-2) [fill=gray] {$d$}; 
      
      \node (1) at (\mx,1.5) [] {1};
      \node (2) at (\mx,0.5) [] {2};
      \node (3) at (\mx,-0.5) [] {3};
      \node (w1) at (\mx,-1.5) [inner sep = 2pt, fill=orange] {$w^\prime$};
      \node (w2) at (\mx,-2.5) [inner sep = 2pt, fill=orange] {$w^{\prime \prime}$};
      
      \def\angle{5}		
      \path[=,thick, double] 	
      (a)		edge [bend right=\angle] (1)
      edge [bend left=\angle] (1)
      (b)	edge [bend left=\angle] (2)
      edge [bend right=\angle] (3)
      (c)		edge (w1)
      edge (w2)
      (d)	    edge [dotted] (w1)
      edge [dotted] (w2);
    \end{tikzpicture}
  \end{minipage}
  \caption{Input graph $G$ (left), graph
    $G^\prime$ resulting by introducing the node $d$ and the edges $e^\prime$ and $e^{\prime\prime}$ 
    (center) and a solution to the
    corresponding \ilp~\eqref{eq_ilp-cardvrap-with-two-scenarios} (right).}
  \label{fig:v-rap-modified-graph-for-two-vulnerable-machines-case}
\end{figure}
Note that the new graph $G^\prime$ remains bipartite.
\\
We will obtain an optimal solution for the \vrpm\ instance $\inst$ from the following \ilp.
\begin{align}
  \label{eq_ilp-cardvrap-with-two-scenarios}
  \begin{split}
    \min~~ & \sum_{e \in E^\prime} x_e \\
    \text{s.t.}~~ & x(\delta(u)) = 2,\; \text{ for each } u \in U^{\prime}, \\
    & x_{e^\prime} = x_{e^{\prime\prime}} = 1, \\
    & x \in \{0,1\}^{E^\prime}.
  \end{split}
\end{align}
Every solution to
\ilp~\eqref{eq_ilp-cardvrap-with-two-scenarios} forms a
collection of cycles of size greater than or equal to two that
cover every node from $U^{\prime}$. Cycles of size two are
formed by parallel edges in $G^\prime$, i.e.\ each such
cycle represents an original edge from $G$.
The cycle covering node $d$ contains the newly introduced edges $e^{\prime}$ and
$e^{\prime\prime}$, and has a size of at least four. This cycle
corresponds to a path from $w^\prime$ to $w^{\prime\prime}$ in $G$
with an even number of edges.
Thus, every solution to
\ilp~\eqref{eq_ilp-cardvrap-with-two-scenarios} defines a union
of a $w^\prime$-$w^{\prime\prime}$-path, a (possibly empty)
matching, some additional even paths, and potentially some further cycles in the original graph $G$.
As $G^\prime$ is bipartite, the constraint matrix of
\ilp~\eqref{eq_ilp-cardvrap-with-two-scenarios} is totally unimodular
(see~\cite{schrijver2002combinatorial} for background on totally unimodular matrices)
implying that \ilp~\eqref{eq_ilp-cardvrap-with-two-scenarios} can be solved in
polynomial time via \lp\ methods. 
\\
Next, we claim that solutions to \ilp~\eqref{eq_ilp-cardvrap-with-two-scenarios} 
correspond to inclusion-wise minimal solutions to the \vrpm\ instance $\inst$ and
vice versa.
For a given feasible solution $x$ of
\ilp~\eqref{eq_ilp-cardvrap-with-two-scenarios}, consider
\begin{align*}
  \rset := \{ w \in W \mid \exists e \in \delta(w)\; \textnormal{ with }\;
  x_e = 1 \}.
\end{align*} 
We argue that $\rset$ is feasible to the original \uvrpm\ instance
$\inst$. To see this, note that each original job node $u$ is adjacent
to at least one non-vulnerable machine node from $X$. 
Each node $u\in U$ located on a cycle in
$G^\prime$ that is induced by $x$ and does not contain a vulnerable node can be easily matched in 
$G[U\dcup X]$ as $u$ is either incident to an isolated edge in
$G[U\dcup X]$ or is part of an even cycle.
Even paths in $G^\prime$ correspond to evens paths in $G$.
Since job nodes are inner nodes of these paths, they can be matched using edges from these paths.
All remaining nodes from $U$ are inner nodes of the even
$w^\prime$-$w^{\prime\prime}$-path induced by $x$, i.e.\ we can 
find a $U$-perfect matching in $G[U\dcup \rset]$ not
using $w^\prime$ and $w^{\prime\prime}$, simultaneously. This shows
that any feasible solution $x$ of
\ilp~\eqref{eq_ilp-cardvrap-with-two-scenarios} corresponds to a
solution $\rset$ feasible to $\inst$.\\
Now, let $\rset$ be any feasible solution to $\inst$. Then, $G[U\dcup
\rset]$ contains a matching $M^\prime$ with $w^\prime\notin V(M^\prime)$ and a
matching $M^{\prime\prime}$ with $w^{\prime\prime}\notin
V(M^{\prime\prime})$. By our assumption that
$\{w^\prime,w^{\prime\prime}\}$ must be contained in every feasible
solution for $\inst$, we have that $w^{\prime\prime}\in V(M^{\prime})$
and $w^{\prime}\in V(M^{\prime\prime})$. 
Then, the symmetric difference $M^\prime\Delta M^{\prime\prime}$ contains an even path $P$ from $w^\prime$ to $w^{\prime\prime}$.
Moreover, $\hat M := M^\prime \setminus E(P)$ is a possibly empty matching on the job nodes not covered by the path $P$.
Consider $x\in\R^{E^\prime}$ with
\begin{align*}
  x_e = \left\{
  \begin{array}{rl}
    1, & \textnormal{ if } e\in\{e^\prime,e^{\prime\prime}\}\cup
         E(P),\\
    1, & \textnormal{ if } e\in \hat M \cup \{\bar e\mid e\in \hat
         M\},\\ 
    0, & \textnormal{ otherwise.}
  \end{array}
  \right.
\end{align*}

By construction the vector $x$ satisfies the constraints of \ilp~\eqref{eq_ilp-cardvrap-with-two-scenarios}.

Finally, solutions of \ilp~\eqref{eq_ilp-cardvrap-with-two-scenarios} can be used to obtain optimal solutions to $\inst$ as follows.
In the graph $G[U \dcup \rset]$ a job node can be adjacent to two non-vulnerable machines (see Figure~\ref{fig:v-rap-modified-graph-for-two-vulnerable-machines-case}).
Such solutions are not optimal and can be identified by checking if
$G[U \dcup \rset]$ has any odd component not containing $w^\prime$ and
$w^{\prime \prime}$.  
In each of these components the number of machine nodes exceeds the number of job nodes by one.
Now by removing an arbitrary machine node from each component we obtain an optimal solution to $\inst$.
\end{proof}

\bigskip
\paragraph{Acknowledgement. }
The work of the second and the third author is part of the Research
Training Group \emph{``Discrete Optimization of Technical Systems
  under Uncertainty''}.
Financial support by the German Research Foundation (DFG) is
gratefully acknowledged through RTG 1855.
Parts of this work were carried out during the second author was visiting IFOR at ETH Zurich. The hospitality is highly appreciated.


\end{document}